\documentclass[a4paper,envcountsame,running heads]{llncs}
\usepackage[utf8]{inputenc}
\usepackage[T1]{fontenc}
\usepackage[english]{babel}
\usepackage{lmodern}
\usepackage{amsmath, amsfonts, amssymb}
\usepackage{graphicx,tikz}
\usepackage{tkz-euclide}
\usetkzobj{all}
\usetikzlibrary{decorations.shapes, calc}
\usepackage{algorithm}
\usepackage[noend]{algpseudocode}
\usepackage{caption}
\usepackage{subcaption}
\usepackage{multirow}
\usepackage{makecell}

\let\doendproof\endproof
\renewcommand\endproof{~\hfill\qed\doendproof}

\spnewtheorem*{remark*}{Remark}{\bfseries}{\itshape}

\let\oldwidehat\widehat
\DeclareRobustCommand{\widehat}[1]{\,\oldwidehat{\mkern-6mu #1 \mkern-4mu}\,}
\newcommand{\NE}{\widehat{\mathit{NE}}}
\newcommand{\tw}{\mathrm{tw}}
\newcommand{\OPT}{\mathit{OPT}}
\newcommand{\polylog}{\mathrm{polylog\;}}
\newcommand{\comment}[1]{}

\usepackage[pdfpagemode=UseNone]{hyperref}
\hypersetup{
     colorlinks=false,
     pdfborder={0 0 0}
}

\title{Near-Linear Query Complexity for Graph Inference}

\author{Sampath Kannan\inst{1}
\and Claire Mathieu\inst{2}
\and Hang Zhou\inst{2}}

\date{}

\institute{
    Department of Computer and Information Science \\
    University of Pennsylvania,
    Philadelphia, PA, USA. \\
    \texttt{kannan@cis.upenn.edu}
\and
    D\'epartement d'Informatique UMR CNRS 8548, \\
    \'Ecole Normale Sup\'erieure,  Paris, France\\
    \texttt{\{cmathieu,hangzhou\}@di.ens.fr}
}

\usepackage[pdfpagemode=UseNone]{hyperref}
\hypersetup{
     colorlinks=false,
     pdfborder={0 0 0}
}

\begin{document}

\maketitle

\begin{abstract}
How efficiently can we find an unknown graph using  distance or shortest path queries between its vertices? Let $G = (V,E)$ be an unweighted, connected graph of bounded degree. The edge set $E$ is initially unknown, and the graph can be accessed using a \emph{distance oracle}, which receives a pair of vertices $(u,v)$ and returns the distance between $u$ and $v$. In the \emph{verification} problem, we are given a hypothetical graph $\hat G = (V,\hat E)$ and want to check whether $G$ is equal to $\hat G$. We analyze a natural greedy algorithm and prove that it uses $n^{1+o(1)}$ distance queries. In the more difficult \emph{reconstruction} problem, $\hat G$ is not given, and the goal is to find the graph $G$. If the graph can be accessed using a \emph{shortest path oracle}, which returns not just the distance but an actual shortest path between $u$ and $v$, we show that extending the idea of greedy gives a reconstruction algorithm that uses $n^{1+o(1)}$ shortest path queries. When the graph has bounded treewidth, we further bound the query complexity of the greedy algorithms for both problems by $\tilde O(n)$. When the graph is chordal, we provide a randomized algorithm for reconstruction using $\tilde O(n)$ distance queries.
\end{abstract}

\section{Introduction}
How efficiently can we find an unknown graph using distance or shortest path queries between its vertices?
This is a natural theoretical question from the standpoint of recovery of hidden information.
This question is related to the \emph{reconstruction} of Internet networks.
Discovering the topology of the Internet is a crucial step for building accurate network models and designing efficient algorithms for Internet applications.
Yet, this topology can be extremely difficult to find, due to the dynamic structure of the network and to the lack of centralized control. 
The network reconstruction problem has been studied extensively~\cite{achlioptas2009bias,Beerliova,dall2006exploring,erlebach2006network,mathieu2013graph,Tarissan:2009:EMC:1719850.1719893}.
Sometimes we have some idea of what the network should be like, based perhaps on its state at some past time, and we want to check whether our image of the network is correct.
This is network \emph{verification} and has received attention recently~\cite{Beerliova,Castro04networktomography,erlebach2006network}.
This is an important task for routing, error detection, or ensuring service-level agreement (SLA) compliance, etc.
For example, Internet service providers (ISPs) offer their customers services that require quality of service (QoS) guarantees, such as voice over IP services, and thus need to check regularly whether the networks are correct.

The topology of Internet networks can be investigated at the router and autonomous system (AS) level, where the set of routers (ASs) and their physical connections (peering relations) are the vertices and edges of a graph, respectively.
Traditionally, we use tools such as traceroute and mtrace to infer the network topology.
These tools generate path information between a pair of vertices.
It is a common and reasonably accurate assumption that the generated path is the shortest one, i.e., minimizes the hop distance between that pair.
In our first theoretical model, we assume that we have access to any pair of vertices and get in return their shortest path in the graph.
Sometimes routers block traceroute and mtrace requests (e.g., due to privacy and security concerns), thus the inference of topology can only rely on delay information.
In our second theoretical model, we assume that we get in return the hop distance between a pair of vertices.
The second model was introduced by Mathieu and Zhou~\cite{mathieu2013graph}.

Graph inference using queries that reveal partial information has been studied extensively in different contexts, independently stemming from a number of applications.
Beerliova et al.~\cite{Beerliova} studied network verification and reconstruction using an oracle, which, upon receiving a node $q$, returns all shortest paths from $q$ to all other nodes, instead of one shortest path between a pair of nodes as in our first model.
Erlebach et al.~\cite{erlebach2006network} studied network verification and reconstruction using an oracle which, upon receiving a node $q$, returns the distances from $q$ to all other nodes in the graph, instead of the distance between a pair of nodes as in our second model.
They showed that minimizing the number of queries for verification is NP-hard and admits an $O(\log n)$-approximation algorithm.
In the {\em network realization} problem, we are given distances between certain pairs of vertices and asked to determine the sparsest graph (in the unweighted case) or the graph of least total weight that realizes these distances. This problem was shown to be NP-hard \cite{CGGS00}.
In evolutionary biology, the goal is to reconstruct evolutionary trees, thus the hidden graph has a tree structure. See for example~\cite{hein1989optimal,King,reyzin2007longest}.
One may query a pair of species and get in return the distance between them in the (unknown) tree.
In our reconstruction problem, we allow the hidden graph to have an arbitrary connected topology, not necessarily a tree structure.

\subsection{The Problem}
Let $G=(V,E)$ be a hidden graph that is connected and unweighted, where $|V|=n$.
We consider two query oracles. 
A \emph{shortest path oracle} receives a pair $(u,v)\in V^2$ and returns a shortest path between $u$ and $v$.
A \emph{distance oracle} receives a pair $(u,v)\in V^2$ and returns the number of edges on a shortest path between $u$ and $v$.

In the \emph{graph reconstruction} problem, we are given the vertex set $V$ and have access to either a distance oracle or a shortest path oracle.
The goal is to find every edge in $E$.

In the \emph{graph verification} problem, again we are given $V$ and have access to one of the two oracles.
In addition, we are given an unweighted, connected graph $\hat G=(V,\hat E)$.
The goal is to check whether $\hat G$ is correct, that is, whether $\hat G=G$.

The efficiency of an algorithm is measured by its \emph{query complexity}\footnote{Expected query complexity in the case of randomized algorithms.}, i.e., the number of queries to an oracle.
We focus on query complexity, while all our algorithms are of polynomial time and space.
We note that $O(n^2)$ queries are enough for both reconstruction and verification using a distance oracle or a shortest path oracle: we only need to query every pair of vertices.

Let $\Delta$ denote the maximum degree of any vertex in the graph $G$.
Unless otherwise stated, we assume that $\Delta$ is bounded, which is reasonable for real networks that we want to reconstruct or verify. 
Indeed, when $\Delta$ is $\Omega (n)$, both reconstruction and verification require $\Omega(n^2)$ distance or shortest path queries, see Section~\ref{sec:lower-bound-unbounded}.\footnote{We note that the $\Omega(n^2)$ lower bound holds even when the graph is restricted to chordal or to bounded treewidth.}

Let us focus on bounded degree graphs.
It is not hard to see that $\Omega(n)$ queries are required.
The central question in this line of work is therefore: {\bf Is the query complexity linear, quadratic, or somewhere in between?}
In~\cite{mathieu2013graph}, Mathieu and Zhou provide a first answer: the query complexity for reconstruction using a distance oracle  is subquadratic: $\tilde{O}(n^{3/2})$.
In this paper, we show that the query complexity for reconstruction using a shortest path oracle or verification using a distance oracle is near-linear: $n^{1+o(1)}$.

\begin{table}[t]
\centering
\caption{Results (for bounded degree graphs). New results are in bold.}
\label{table:results}
\begin{tabular}{|l|l|}
  \hline
   {\em Objective} & {\em Query complexity}\\
  \hline
  \Gape[7pt]{verification using a distance oracle} &
  \multirowcell{2}[0pt][l]{
    \;\bfseries\boldmath$n^{1+o(1)}$\\
    \;\bfseries bounded treewidth:  $\tilde O(n)$ \\
    \;\bfseries (Thm~\ref{thm:verify-dist} and Thm~\ref{thm:recons-sp})
  }\\
  %
  \cline{1-1}
  \Gape[8pt]{reconstruction using a shortest path oracle\;} &\\
  \hline
  reconstruction using a distance oracle &
    \Gape[5pt]{\makecell[l]{
      \;$\tilde O(n^{3/2})$\quad\cite{mathieu2013graph} \\ 
      \;{\bfseries\boldmath$\Omega(n\log n/\log\log n)$\quad(Thm~\ref{thm:information-lower-bound})} \\
      \;outerplanar: $\tilde {O}(n)$\quad\cite{mathieu2013graph}\\
      \;{\bfseries\boldmath chordal: $\tilde {O}(n)$\quad(Thm~\ref{thm:recons-chordal})}}
    }\\
  \hline
\end{tabular}
\end{table}

\subsection{Our Results}
\label{sec:results}
\subsubsection{Verification.}
\begin{theorem}
\label{thm:verify-dist}
For graph verification using a distance oracle, there is a deterministic algorithm (Algorithm~\ref{alg:verif-greedy}) with query complexity $n^{1+O\left(\sqrt{(\log \log n+\log \Delta)/\log n}\right)}$, which is $n^{1+o(1)}$ when the maximum degree $\Delta=n^{o(1)}$. 
If the graph has treewidth $\tw$, the query complexity can be further bounded by $O(\Delta(\Delta+\tw\log n)n\log^2 n)$, which is $\tilde{O}(n)$ when $\Delta$ and $\tw$ are $O(\polylog n)$. 
\end{theorem}

The main task for verification is to confirm the \emph{non-edges} of the graph.
Algorithm~\ref{alg:verif-greedy} is greedy: every time it makes a query that confirms the largest number of non-edges that are not yet confirmed.
To analyze the algorithm, first, we show that its query complexity is $O(\log n)$ times the optimal number of queries $\OPT$ for verification.
This is based on a reduction to the \textsc{Set-Cover} problem, see Section~\ref{sec:verif-set-cover}.
It only remains to bound $\OPT$.

To bound $\OPT$ and get the first statement in Theorem~\ref{thm:verify-dist}, it is enough to prove the desired bound for a different verification algorithm.
This algorithm is a more sophisticated recursive version of the algorithm in~\cite{mathieu2013graph}. 
Recursion is a challenge because, when we query the pair $(u,v)$ in a recursive subgraph, the oracle returns the distance between $u$ and $v$ in the entire graph, not just within the subgraph.
Thus new ideas are introduced for the algorithmic design.
See Section~\ref{sec:proof-verify-bound-OPT}.
\comment{
As in~\cite{mathieu2013graph}, this algorithm selects a set of \emph{centers} partitioning $V$ into Voronoi cells and expands them slightly so as to cover all edges of $E$.
But unlike~\cite{mathieu2013graph}, instead of using exhaustive search inside each cell, the algorithm verifies each cell recursively.
The recursion is non-trivial because the distance oracle returns the distance in the original graph, not in the cell.
Straightforward attempts to use recursion lead either to subcells that do not cover every edge of the cell, or to excessively large subcells.
Our approach is to allow selection of centers \emph{outside} the cell, while still limiting the subcells to being contained \emph{inside} the cell (Figure~\ref{fig:two-levels}).
This simple but subtle setup is the main novelty of the algorithmic design.
In terms of analysis, the main technical novelty consists in Lemma~\ref{lem:Da-same}.
}

To show the second statement in Theorem~\ref{thm:verify-dist}, similarly, we design another recursive verification algorithm with query complexity $\tilde O(n)$ for graphs of bounded treewidth.
The algorithm uses some bag of a tree decomposition to separate the graph into balanced subgraphs, and then recursively verifies each subgraph. The same obstacle to recursion occurs.
Our approach here is to add a few weighted edges to each subgraph in order to preserve the distance metric.
See Section~\ref{sec:verify-bounded-width}.

We note that each query to a distance oracle can be simulated by the same query to a shortest path oracle.
So from Theorem~\ref{thm:verify-dist}, we have:

\begin{corollary}
\label{cor:verify-sp}
For graph verification using a shortest path oracle, Algorithm~\ref{alg:verif-greedy} achieves the same query complexity as in Theorem~\ref{thm:verify-dist}.
\end{corollary}

\subsubsection{Reconstruction.}

\begin{theorem}
\label{thm:recons-sp}
For  graph reconstruction using a shortest path oracle, there is a deterministic algorithm (Algorithm~\ref{alg:recons-greedy}) that achieves the same query complexity as in Theorem~\ref{thm:verify-dist}.
\end{theorem}

The key is to formulate this problem as a problem of verification using a distance oracle, so that we get the same query complexity as in Theorem~\ref{thm:verify-dist}. 
We extend ideas of the greedy algorithm in Theorem~\ref{thm:verify-dist} to design Algorithm~\ref{alg:recons-greedy}, and we show that  each query to a shortest path oracle makes as much progress for reconstruction as the corresponding query to a distance oracle would have made for verifying a given graph.
The main realization here is that reconstruction can be viewed as the verification of a dynamically changing graph.
See Section~\ref{sec:recon}.

\begin{theorem}
\label{thm:recons-chordal}
For reconstruction of {\bf chordal graphs} using a distance oracle,  there is a randomized algorithm (Algorithm~\ref{alg:recons-chordal}) with query complexity $O\big(\Delta^3 2^\Delta\cdot\allowbreak n(2^\Delta+\log^2n)\log n\big)$, which is $\tilde O(n)$ when the maximum degree $\Delta$ is $O(\log \log n)$.
\end{theorem}

The algorithm first finds a separator using random sampling and statistical estimates, as in~\cite{mathieu2013graph}.
Then it partitions the graph into subgraphs with respect to this separator and recurses on each subgraph.
However, the separator here is a clique instead of an edge in~\cite{mathieu2013graph} for outerplanar graphs.
Thus the main difficulty is to design and analyze a more general tool for partitioning the graph, see Section~\ref{sec:partition}.
The reconstruction algorithm is in Section~\ref{sec:recons-chordal-algo}.

On the other hand, graph reconstruction using a distance oracle has a lower bound that is slightly higher than trivial $\Omega(n)$ bound, as in the following theorem. Its proof is in Section~\ref{sec:information-lower-bound}.

\begin{theorem}
\label{thm:information-lower-bound}
For graph reconstruction using a distance oracle, assuming the maximum degree $\Delta\geq 3$ is such that $\Delta=o\left(n^{1/2}\right)$, any algorithm has query complexity $\Omega(\Delta n\log n/\log\log n)$.
\end{theorem}

It is an outstanding open question whether there is a reconstruction algorithm using a near-linear number of queries to a distance oracle for degree bounded graphs in general.

\section{Notation and Preliminaries}
\label{sec:preli-decompo}
Let $\delta$ be the distance metric of $G$.
For a subset of vertices $S\subseteq V$ and a vertex $v\in V$, define $\delta (S,v)$ to be $\min _ { s \in S} \delta (s,v)$.
For $v\in V$,  let $N(v)=\{u\in V: \delta(u,v)\leq 1\}$ and let $N_2(v)=\{u\in V: \delta(u,v)\leq 2\}$.
For $S\subseteq V$,  let $N(S)=\bigcup_{s\in S} N(s)$.
We define $\hat \delta$, $\hat N$, and $\hat N_2$ similarly with respect to the graph $\hat G$.

A pair of vertices $\{u,v\}\subseteq V$ is called a \emph{non-edge} of the graph $G=(V,E)$ if $\{u,v\}\notin E$.

For a subset of vertices $S\subseteq V$, let $G[S]$ be the subgraph induced by $S$.
For a subset of edges $H\subseteq E$, we identify $H$ with the subgraph induced by the edges of $H$.
Let $\delta_H$ denote the distance metric of the subgraph $H$.

For a vertex $s\in V$ and a subset $T\subseteq V$, define \textsc{Query}$(s,T)$ as \textsc{Query}$(s,t)$ for every $t\in T$.
For subsets $S,T\subseteq V$, define \textsc{Query}$(S,T)$ as \textsc{Query}$(s,t)$ for every $(s,t)\in S\times T$.

\begin{definition}
A subset $S\subseteq V$ is a \emph{$\beta$-balanced separator} of the graph $G=(V,E)$ (for $\beta<1$) if the size of every connected component of $G\setminus S$ is at most $\beta |V|$.
\end{definition}

\begin{definition}
A {\em tree decomposition} of a graph $G = (V,E)$ is a tree $T$ with nodes $n_1, n_2, \ldots , n_\ell$. Node
$n_i$ is identified with a \emph{bag} $S_i \subseteq V$, satisfying the following conditions:
\begin{enumerate}
\item For every vertex $v$ in $G$, the nodes whose bags contain $v$ form a connected subtree of $T$.
\item For every edge $(u,v)$ in $G$, some bag contains both $u$ and $v$.
\end{enumerate}
The {\em width} of the decomposition is the size of the largest bag minus 1, and the {\em treewidth} of $G$ is the minimum width over all possible tree decompositions of $G$.
\end{definition}

\begin{lemma}[\cite{reed2003algorithmic}]
\label{lemma:exist-balanced}
Let $G$ be a graph of treewidth $k$.
Any tree decomposition of width $k$ contains a bag $S$ that is a $(1/2)$-balanced separator of $G$.
\end{lemma}

A graph is \emph{chordal} if every cycle of length greater than three has a chord: namely, an edge connecting two nonconsecutive vertices on the cycle.
An introduction to chordal graphs can be found in {\em e.g.,}~\cite{blair1993introduction}.

\begin{lemma}[\cite{blair1993introduction}]
\label{lemma:chordal}
Let $G$ be a chordal graph.
Then $G$ has a tree decomposition where every bag is a maximal clique\footnote{A maximal clique is a clique which is not contained in any other clique.} and every maximal clique appears exactly once in this decomposition.
\end{lemma}

From Lemmas~\ref{lemma:exist-balanced} and~\ref{lemma:chordal}, we have:

\begin{corollary}
\label{cor:exist-clique-balanced}
Let $G$ be a chordal graph of maximum degree $\Delta$.
Then $G$ has treewidth at most $\Delta$, and there exists a clique $S\subseteq V$ of size at most $\Delta+1$ that is a $(1/2)$-balanced separator of $G$.
\end{corollary}

\section{Proof of Theorem~\ref{thm:verify-dist}}

\subsection{Greedy Algorithm}
\label{sec:verif-set-cover}

The task of verification comprises verifying that every edge in $\hat G$ is an edge in $G$, and verifying that every non-edge of $\hat G$ is a non-edge of $G$.
The second part, called \emph{non-edge verification}, is the main task for graphs of bounded degree.\footnote{In non-edge verification, we always assume that $\hat E\subseteq E$.}

\begin{theorem}
\label{thm:verify-set-cover}
For graph verification using a distance oracle, there is a deterministic greedy algorithm (Algorithm~\ref{alg:verif-greedy}) that uses at most $\Delta n+(\ln n+1)\cdot \OPT$ queries, where $\OPT$ is the optimal number of queries for non-edge verification.
\end{theorem}

Now we prove Theorem~\ref{thm:verify-set-cover}.
Let $\NE$ be the set of the non-edges of $\hat G$.
For each pair of vertices $(u,v)\in V^2$, we define $S_{u,v}\subseteq \NE$ as follows:
\begin{equation}
\label{eqn:Suv}
S_{u,v}=\left\{ \{a,b\}\in \NE: \hat{\delta}(u,a) + \hat{\delta}(b,v) + 1 < \hat{\delta} (u,v) \right\}.
\end{equation}
The following two lemmas relate the sets $S_{u,v}$ with non-edge verification.

\begin{lemma}
\label{lem:Suv}
Assume that $\hat E\subseteq E$.
Let $(u,v)\in V^2$ be such that $\delta(u,v)=\hat \delta(u,v)$.
Then every pair $\{a,b\}\in S_{u,v}$ is a non-edge of $G$.
\end{lemma}

\begin{proof}
Let $\{a,b\}$ be any pair in $S_{u,v}$.
By the triangle inequality, ${\delta}(u,a) + \delta(a,b) + {\delta}(b,v)   \geq \delta(u,v)=\hat \delta(u,v)$.
By the definition of $S_{u,v}$ and using $\hat E\subseteq E$, we have $\hat{\delta}(u,v)> \hat{\delta}(u,a) +  \hat{\delta}(b,v)+1\geq {\delta}(u,a) + {\delta}(b,v) +1$.
Thus $\delta(a,b)>1$, i.e., $\{a,b\}$ is a non-edge of $G$.
\end{proof}

\begin{lemma}
\label{lem:cover}
If a set of queries $T$ verifies that every non-edge of $\hat G$ is a non-edge of $G$, then $\bigcup_{(u,v)\in T} S_{u,v}=\NE$.
\end{lemma}

\begin{proof}
Assume, for a contradiction, that some $\{a,b\}\in \NE$ does not belong to any $S_{u,v}$ for $(u,v)\in T$.
Consider adding $\{a,b\}$ to the set of edges of $\hat E$: this will not create a shorter path between $u$ and $v$, for any $(u,v)\in T$.
Thus including $\{a,b\}$ in $\hat{E}$ is consistent with the answers of all queries in $T$. 
This contradicts the assumption that $T$ verifies that $\{a,b\}$ is a non-edge of $G$.
\end{proof}

From Lemmas~\ref{lem:Suv} and~\ref{lem:cover}, the non-edge verification is equivalent to the \textsc{Set-Cover} problem with the universe $\NE$ and the sets $\{S_{u,v}:(u,v)\in V^2\}$.
The \textsc{Set-Cover} instance can be solved using the well-known
greedy algorithm~\cite{johnson1974approximation}, which gives a $(\ln n +1)$-approximation.
Hence our greedy algorithm for verification (Algorithm~\ref{alg:verif-greedy}).
For the query complexity, first, verifying that $\hat E\subseteq E$ takes at most $\Delta n$ queries, since the graph has maximum degree $\Delta$.
The part of non-edge verification uses a number of queries that is at most $(\ln n+1)$ times the optimal number of queries.
This proves Theorem~\ref{thm:verify-set-cover}.

\begin{algorithm}
\caption{Greedy Verification}
\label{alg:verif-greedy}
\begin{algorithmic}[1]
\Procedure{Verify}{$\hat G$}
    \For{$\{u,v\}\in \hat E$}
        \textsc{Query}$(u,v)$
            \EndFor
    \If{some $\{u,v\}\in\hat E$ has $\delta(u,v)\neq \hat\delta (u,v)$} \textbf{return} \emph{no} \EndIf
    \State $Y\gets\emptyset$  
    \While{$\hat E \cup Y$ does not cover all vertex pairs}
        \State choose  $(u,v)$ that maximizes $|S_{u,v}\setminus Y |$\Comment{$S_{u,v}$ defined in Equation~\eqref{eqn:Suv}}
        \State Query($u,v$) 
        \If{$\delta(u,v)=\hat\delta (u,v)$}
            \State $Y\gets Y\cup S_{u,v}$       
        \Else
            \State \textbf{return} \emph{no}
        \EndIf
    \EndWhile
    \State \textbf{return} \emph{yes}
\EndProcedure
\end{algorithmic}
\end{algorithm}

\subsection{Bounding $\OPT$ to Prove Theorem~\ref{thm:verify-dist}}
\label{sec:verif}
From Theorems~\ref{thm:verify-set-cover}, in order to prove Theorem~\ref{thm:verify-dist}, we only need to bound $\OPT$, as in the following two theorems.

\begin{theorem}
\label{thm:verify-bound-OPT}
For graph verification using a distance oracle, the optimal number of queries $\OPT$ for non-edge verification is 
$n^{1+O\left(\sqrt{(\log \log n+\log \Delta)/\log n}\right)}.$
\end{theorem}

\begin{theorem}
\label{thm:verify-bound-OPT-bounded-width}
For graph verification using a distance oracle, if the graph has treewidth $\tw$, then the optimal number of queries $\OPT$ for non-edge verification is $O(\Delta(\Delta+\tw\log n)n\log n)$.
\end{theorem}

Theorem~\ref{thm:verify-dist} follows trivially from Theorems~\ref{thm:verify-set-cover},~\ref{thm:verify-bound-OPT}, and~\ref{thm:verify-bound-OPT-bounded-width}, by noting that both $\Delta$ and $\log n$ are smaller than $n^{\sqrt{(\log \log n+\log \Delta)/\log n}}$.
The proof of Theorem~\ref{thm:verify-bound-OPT} is in Section~\ref{sec:proof-verify-bound-OPT}, and the proof of Theorem~\ref{thm:verify-bound-OPT-bounded-width} is in Section~\ref{sec:verify-bounded-width}.

\subsection{Proof of Theorem~\ref{thm:verify-bound-OPT}}
\label{sec:proof-verify-bound-OPT}
To show Theorem~\ref{thm:verify-bound-OPT}, we provide a recursive algorithm for non-edge verification with the query complexity in the theorem statement.
As in~\cite{mathieu2013graph}, the algorithm selects a set of \emph{centers} partitioning $V$ into Voronoi cells and expands them slightly so as to cover all edges of $G$.
But unlike~\cite{mathieu2013graph}, instead of using exhaustive search inside each cell, the algorithm verifies each cell recursively.
The recursion is a challenge because the distance oracle returns the distance in the entire graph, not in the cell.
Straightforward attempts to use recursion lead either to subcells that do not cover every edge of the cell, or to excessively large subcells.
To make the recursion work, we allow selection of centers \emph{outside} the cell, while still limiting the subcells to being contained \emph{inside} the cell (Figure~\ref{fig:two-levels}).
This simple but subtle setup is one novelty of the algorithmic design.

Let $U\subseteq V$ represents the set of vertices for which we are currently verifying the induced subgraph.
The goal is to verify that every non-edge of $\hat G[U]$ is a non-edge of $G[U]$.
This is equivalent to verifying that every edge of $G[U]$ is an edge of $\hat G[U]$.

The algorithm uses a subroutine to find \emph{centers} $A\subseteq V$ such that the vertices of $U$ are roughly equipartitioned into the Voronoi cells centered at vertices in $A$.
For a set of centers $A\subseteq V$ and a vertex $w\in V$, let $\hat C_A(w)=\{v\in V:\hat \delta(w,v)<\hat \delta(A,v)\}$, which represents the Voronoi cell of $w$ if $w$ is added to the set of centers.
We note that $\hat C_A(w)=\emptyset$ for $w\in A$, since in that case, $\hat \delta(w,v)\geq\hat \delta(A,v)$ for every $v\in V$.
The subscript $A$ is omitted when clear from the context.

\begin{lemma}
\label{lem:subset-center}
Given a graph $\hat G=(V,\hat E)$, a subset of vertices $U\subseteq V$, and an integer $s\in[1,n]$, Algorithm~\ref{alg:subset-center} computes a subset of vertices $A\subseteq V$, such that:
\begin{itemize}
\item the expected size of the set $A$ is at most $2s\log n$; and
\item for every vertex $w\in V$, we have $|\hat C_A(w)\cap U|\leq 4|U|/s$.
\end{itemize}
\end{lemma}

\begin{algorithm}
\caption{Finding Centers for a Subset}
\label{alg:subset-center}
\begin{algorithmic}[1]
\Function{Subset-Centers}{$\hat G, U, s$}
    \State $A\gets\emptyset$
    \While{there exists $w\in V$ such that $|\hat C(w)\cap U|>4|U|/s}$
        \State $W\gets \{w\in V: |\hat C(w)\cap U|>4|U|/s\}$
        \State Add  each element of $W$ to $A$ with probability $\min\left(s/|W|,1\right)$
    \EndWhile
    \State \textbf{return} $A$
\EndFunction
\end{algorithmic}
\end{algorithm}

Algorithm~\ref{alg:subset-center} is a generalization of the algorithm \textsc{Center} in~\cite{thorup2001compact}; and Lemma~\ref{lem:subset-center} is a trivial extension of Theorem~3.1 in~\cite{thorup2001compact}.\footnote{As noted in~\cite{thorup2001compact}, it is possible to derandomize the center-selecting algorithm, and its running time is still polynomial.}

Using a set of centers $A$, we define, for each $a\in A$, its \emph{extended Voronoi cell} $\hat D_a\subseteq U$ as follows:
\begin{equation}
\label{eqn:def-Da}
\hat D_a=\left(\bigcup\left\{\hat C(b): b \in \hat N_2(a)\right\}\cup \hat N_2(a)\right) \cap U.
\end{equation}
We define $C(w)$ and $D_a$ similarly as $\hat C(w)$ and $\hat D_a$, but with respect to the graph $G$.

The following lemma is the base of the recursion.
Its proof is similar to that of Lemma~3 in~\cite{mathieu2013graph}.
\begin{lemma}
\label{lem:edge-cover}
$\bigcup_{a\in A} G[D_a]$ covers every edge of $G[U]$.
\end{lemma}
\begin{proof}
We prove that for every edge $\{u,v\}$ of $G[U]$, there is some $a\in A$, such that both $u$ and $v$ are in $D_a$.
Let $\{u,v\}$ be any edge of $G[U]$. Without loss of generality, we assume $\delta(A,u)\leq \delta(A,v)$.
We choose $a\in A$ such that $\delta(a,u)=\delta(A,u)$.
If $\delta(a,u)\leq 1$, then both $u$ and $v$ are in $N_2(a)\cap U\subseteq D_a$.
If $\delta(a,u)\geq 2$, let $b$ be the vertex at distance 2 from $a$ on a shortest $a$-to-$u$ path in $G$.
By the triangle inequality, we have $\delta(b,v)\leq \delta(b,u)+\delta(u,v)=\delta(b,u)+1$.
Since $\delta(b,u)=\delta(a,u)-2$ and $\delta(a,u)=\delta(A,u)\leq \delta(A,v)$, we have $\delta(b,u)< \delta(A,u)$ and $\delta(b,v)<\delta(A,v)$. So both $u$ and $v$ are in $C(b)\cap U$, which is a subset of $D_a$ since $b\in N_2(a)$.
\end{proof}


From Lemma~\ref{lem:edge-cover}, verifying that every edge of $G[U]$ is an edge of $\hat G[U]$ reduces to verifying that every edge of $G[D_a]$ is an edge of $\hat G[D_a]$ for every $D_a$. 
To see this, consider any edge $\{u,v\}$ of $G[U]$.
There exists $a\in A$ such that $u,v\in D_a$.
It is enough to verify that $\{u,v\}$ is an edge of $\hat G[D_a]$, hence an edge of $\hat G[U]$.
This observation enables us to apply recursion on each $D_a$.

The main difficulty is: {\bfseries\boldmath How to obtain $D_a$ efficiently?}
If we compute $D_a$ from its definition, we first need to compute $N_2(a)$, which takes too many queries since $N_2(a)$ may contain nodes outside $U$.
Instead, a careful analysis shows that we can check whether $D_a=\hat D_a$ without even knowing $N_2(a)$, and $\hat D_a$ can be inferred from the graph $\hat G$ with no queries.
This is shown in Lemma~\ref{lem:Da-same}, which is the main novelty of the algorithmic design.

\begin{lemma}
\label{lem:Da-same}
Assume that $\hat E\subseteq E$.
If $\delta(u,v)=\hat\delta(u,v)$ for every pair $(u,v)$ from $\bigcup_{a\in A} \hat N_2(a)\times U$, then $D_a=\hat D_a$ for all $a \in A$.
\end{lemma}

\begin{proof}
The proof is delicate but elementary.
For every $b\in \bigcup_{a\in A}\hat N_2(a)$, we have $\hat C(b)\cap U=C(b)\cap U$, because we have verified that $\hat \delta(b,u)=\delta(b,u)$ and $\hat \delta(A,u)=\delta(A,u)$ for every $u\in U$.
Therefore, $\hat D_a$ can be rewritten as $\left(\bigcup\left\{ C(b): b \in \hat N_2(a)\right\}\cup \hat N_2(a)\right) \cap U.$
Since $\hat E\subseteq E$, we have $\hat N_2(a)\subseteq N_2(a)$.
Therefore $\hat D_a\subseteq D_a$.

On the other hand, we have $N_2(a)\cap U\subseteq\hat N_2(a)\cap U$, because we have verified that $\hat \delta(a,u)=\delta(a,u)$ for all $u$ in $N_2(a)\cap U$.
To prove $D_a\subseteq \hat D_a$, it only remains to show that, for any vertex $u\notin N_2(a)$ such that $u\in C(b)\cap U$ for some $b\in N_2(a)$, we have $u\in C(x)\cap U$ for some $x\in \hat N_2(a)$.
We choose $x$ to be the vertex at distance 2 from $a$ on a shortest $a$-to-$u$ path in $\hat G$.
By the assumption and the definition of $x$, we have: $\delta(x,u)=\hat \delta(x,u)=\hat \delta(a,u)-2=\delta(a,u)-2$.
By the triangle inequality, and using $b\in N_2(a)$ and $u\in C(b)$, we have: $\delta(a,u)\leq \delta(a,b)+\delta(b,u)\leq 2+\delta(b,u)< 2+\delta(A,u)$.
Therefore $\delta(x,u)< \delta(A,u)$.
Thus $u\in C(x)\cap U$.
\end{proof}



The recursive algorithm for non-edge verification is in Algorithm~\ref{alg:verify-compact-routing}.
It queries every $(u,v)\in \bigcup_{a\in A} \hat N_2(a)\times U$ and then recurses on each extended Voronoi cell $\hat D_a$.
See Figure~\ref{fig:two-levels}.
It returns \emph{yes} if and only if every query during the execution gives the right distance.
The parameters $n_0$ and $s$ are defined later.
We assume that every edge of $\hat G$ has already been confirmed, i.e., $\hat E\subseteq E$.
Correctness of the algorithm follows trivially from Lemmas~\ref{lem:edge-cover} and~\ref{lem:Da-same}.

\begin{algorithm}
\caption{Recursive Verification}
\label{alg:verify-compact-routing}
\begin{algorithmic}[1]
\Procedure{Verify-Subgraph}{$\hat G,U$}
    \If{$|U|>n_0$}  
        \State $A\gets \textsc{Subset-Centers}(\hat G, U, s)$ \Comment{Algorithm~\ref{alg:subset-center}}\label{line:A}
        \For{$a\in A$}
            \State \textsc{Query}$(\hat N_2(a),U)$\label{line:check}
            \State \textsc{Verify-Subgraph}$(\hat G, \hat D_a)$\Comment{$\hat D_a$ defined in Equation~\eqref{eqn:def-Da}}\label{line:rec}
        \EndFor
    \Else \;\textsc{Query}$(U,U)$
    \EndIf
\EndProcedure
\end{algorithmic}
\end{algorithm}

\newcommand{\sz}{}
\begin{figure}[t]
\begin{minipage}[c]{0.48\textwidth}
\centering
\begin{tikzpicture}[scale=0.32]
\tkzInit[ymin=-12,ymax=9,xmin=-10,xmax=10]
\tkzClip
\tkzDefPoint(0,0){O}
\tkzLabelPoint[right=0.1](O){\sz $a$}
\node at (1.5,1.5){\sz $\hat D_a$};
\tkzDefPoint(5,5){A}
\tkzDefPoint(-5,0){B}
\tkzDefPoint(-2,-4.8){C}
\tkzDefPoint(3,-6){D}
\tkzDefPointBy[reflection=over A--B](O) \tkzGetPoint{E}
\tkzDefPointBy[reflection=over B--C](O) \tkzGetPoint{F}
\tkzDefPointBy[reflection=over C--D](O) \tkzGetPoint{G}
\tkzDefPointBy[reflection=over D--A](O) \tkzGetPoint{H}
\tkzDrawPoints[size=10](O,E,F,G,H)
\tkzDefLine[mediator](E,F) \tkzGetSecondPoint{I}
\tkzDefLine[mediator](F,G) \tkzGetSecondPoint{J}
\tkzDefLine[mediator](G,H) \tkzGetSecondPoint{K}
\tkzDefLine[mediator](H,E) \tkzGetSecondPoint{L}
\foreach \i/\j in {
    A/B,B/C,C/D,D/A,A/L,B/I,C/J,D/K}{\tkzDrawSegment[dotted](\i,\j)
}
\foreach \i/\j in {A/B,B/C,C/D,D/A}{
    \tkzDefPointWith[orthogonal normed,K=-1](\i,\j) \tkzGetPoint{\i\j}
    \tkzDefPointWith[orthogonal normed](\j,\i) \tkzGetPoint{\j\i}
}
\tikzstyle{mystyle}=[color=black,line width=2pt]
\foreach \i/\j in {A/B,B/C,C/D,D/A}{
    \tkzDrawSegment[mystyle](\i\j,\j\i)
}
\tkzDrawArc[mystyle](B,BA)(BC)
\tkzDrawArc[mystyle](C,CB)(CD)
\tkzDrawArc[mystyle](D,DC)(DA)
\tkzDrawArc[mystyle](A,AD)(AB)
\tkzDefPoint(5,-5){M}
\tkzLabelPoint[right=0.1](M){\sz $a'$}
\tkzDefPoint(3,2.5){P}
\tkzDefPoint(-3.8,-3){Q}
\tkzDefPoint(-5.5,2){R}
\tkzDefPoint(-3,-6.5){N}
\tkzDrawPoints[size=2,fill=white,color=red](M,N,P,Q,R)
\foreach \i/\j in {M/N,M/Q,M/P,P/R,R/Q,P/Q,Q/N}{
    \tkzDefLine[mediator](\i,\j) \tkzGetPoints{\i\j1}{\i\j2}
}
\tkzInterLL(MN1,MN2)(MQ1,MQ2) \tkzGetPoint{S}
\tkzInterLL(MP1,MP2)(MQ1,MQ2) \tkzGetPoint{T}
\tkzInterLL(PQ1,PQ2)(RQ1,RQ2) \tkzGetPoint{U}
\tkzInterLL(T,MP2)(AD,DA) \tkzGetPoint{X1}
\tkzInterLL(S,MN1)(CD,DC) \tkzGetPoint{X2}
\tkzInterLL(S,QN2)(BC,CB) \tkzGetPoint{X3}
\tkzInterLL(U,RQ2)(BC,CB) \tkzGetPoint{X4}
\tkzInterLL(U,PR2)(AB,BA) \tkzGetPoint{X5}
\foreach \i/\j in {S/T,U/T,S/X2,S/X3,T/X1,U/X5,U/X4}{
    \tkzDrawSegment[dashed, red](\i,\j)
}
\foreach \i/\j in {MP2/T,T/S,S/MN1}{
    \tkzDefPointWith[orthogonal normed,K=-1](\i,\j) \tkzGetPoint{\i\j}
    \tkzDefPointWith[orthogonal normed](\j,\i) \tkzGetPoint{\j\i}
}
\tikzstyle{mystyle}=[color=red,line width=1pt]
\tkzInterLL(MP2T,TMP2)(AD,DA) \tkzGetPoint{V}
\tkzInterLL(SMN1,MN1S)(CD,DC) \tkzGetPoint{W}
\tkzDrawSegment[mystyle](V,TMP2)
\tkzDrawSegment[mystyle](TS,ST)
\tkzDrawSegment[mystyle](W,SMN1)
\tkzDrawArc[mystyle](T,TMP2)(TS)
\tkzDrawArc[mystyle](S,ST)(SMN1)
\tkzDrawSegment[mystyle](V,DA)
\tkzDrawSegment[mystyle](W,DC)
\tkzDrawArc[mystyle](D,DC)(DA)
\node at (2,-4){\sz $\hat D'_{a'}$};
\end{tikzpicture}
\end{minipage}
\hfill
\begin{minipage}[c]{0.48\textwidth}
The solid points are top-level centers returned by \textsc{Subset-Centers}$(\hat G,V,s)$.
 The dotted lines indicate the partition of $V$ into Voronoi cells by those centers.
 For a center $a$, expanding slightly its Voronoi cell results in $\hat D_a$ (the region inside the outer closed curve).
 On the second level of the recursive call for $\hat D_a$, the hollow points are the centers returned by \textsc{Subset-Centers}$(\hat G,\hat D_a,s)$.
 Observe that some of those centers lie outside $\hat D_a$.
 The dashed lines indicate the partition of $\hat D_a$ into Voronoi cells by those centers.
 Similarly, for a center $a'$, expanding slightly its Voronoi cell results in $\hat D'_{a'}$ (the region inside the inner closed curve).
 Note that every $\hat D'_{a'}$ is inside $\hat D_a$.
 \end{minipage}
\caption{\footnotesize Two levels of recursive calls of \textsc{Verify-Subgraph}$(\hat G,V)$}
 \label{fig:two-levels}
\end{figure}

Next, we analysis the query complexity of \textsc{Verify-Subgraph}$(\hat G,V)$.
\comment{
Here, to provide intuition, we assume that $\Delta=O(1)$ and we bound the query complexity by $\tilde O(n^{4/3})$.

Let $s=n^{1/3}$ and $n_0=\left(4(\Delta^2+1)\right)^3$ be the parameters in \textsc{Verify-Subgraph}.
Consider any recursive call \textsc{Verify-Subgraph}$(\hat G, U)$ where $|U|>n_0$.
Let $A\subseteq V$ be the centers returned by \textsc{Subset-Centers}.
For every $a\in A$, $\hat N_2(a)$ has size at most $\Delta^2+1$, since the graph has maximum degree $\Delta$.
By Lemma~\ref{lem:subset-center}, every $|\hat C(w)\cap U|$ is at most $4|U|/s$.
So the size of every $\hat D_a$ is at most $(\Delta^2+1)\cdot 4|U|/s$.
Therefore by induction, any subset $U$ on the $k^{\rm th}$ level of the recursion has size at most $n\cdot\left(4(\Delta^2+1)/s\right)^{k-1}$.
From the definition of $n_0$ and $s$, we know that there are at most 4 levels.
Consider a subset $U$ on the $k^{\rm th}$ level of the recursion.
If $|U|\leq n_0$, the number of queries is $|U|^2=O(1)$; otherwise the number of queries is $O(|A|\cdot |U|)$, where $|A|\leq 2s\log n$ by Lemma~\ref{lem:subset-center}.
By induction, the number of recursive calls on the $k^{\rm th}$ level of the recursion is at most $(2s\log n)^{k-1}$, so the overall query complexity on this level is $O(n s\log^{k}n)$.
Therefore, the overall query complexity over all of the 4 levels is $O(n s\log^4 n)$, which is $\tilde O(n^{4/3})$.

Now we give the full proof of the complexity stated in Theorem~\ref{thm:verify-bound-OPT}.
}
Define \[k_0=\left\lfloor \sqrt{\frac{\log n}{\log \left(\log n\cdot32(\Delta^2+1)^2\right)}}\right\rfloor.\]
Let $s=n^{1/k_0}$ and $n_0=\left(4(\Delta^2+1)\right)^{k_0}$ be the parameters in \textsc{Verify-Subgraph}.
Consider any recursive call when $|U|>n_0$.
Let $A\subseteq V$ be the centers returned by \textsc{Subset-Centers}.
By Lemma~\ref{lem:subset-center}, $|A|\leq 2s\log n$ and every $|\hat C(w)\cap U|$ is at most $4|U|/s$.
Since the graph has maximum degree $\Delta$, the size of every $\hat D_a$ is at most $(\Delta^2+1)\cdot \max(4|U|/s,1)$.
Therefore by induction, for any $1\leq k\leq k_0+1$, any subset $U$ on the $k^\mathrm{th}$ level of the recursion has size at most $t_k:=n\left(4(\Delta^2+1)/s\right)^{k-1}$, where $t_{k_0+1}=n_0$. 
Hence the maximum level of the recursion is at most $k_0+1$.

First, consider the recursive calls with $|U|\leq n_0$.
There are at most $(2s\log n)^{k_0}$ such calls and each takes $|U|^2\leq \left(4(\Delta^2+1)\right)^{2k_0}$ queries.
So their overall query complexity is at most $n\cdot\left(\log n \cdot 32(\Delta^2+1)^2\right)^{k_0}\leq n^{1+1/k_0}$.

Next, consider the recursive calls with $|U|>n_0$ on the $k^\mathrm{th}$ level of the recursion for some fixed $k\in[1,k_0]$.\footnote{We note that there are no recursive calls on the $(k_0+1)^\mathrm{th}$ level (i.e., last level) of the recursion with $|U|>n_0$.}
There are at most $(2s\log n)^{k-1}$ such calls and each takes at most $(\Delta^2+1)|A|\cdot |U|$ queries, where $|U|\leq t_k$.
So their overall query complexity is at most $n^{1+1/k_0}\left(\log n\cdot 8(\Delta^2+1)\right)^{k}$.
Summing over $k$ from 1 to $k_0$, the query complexity of all recursive calls with $|U|>n_0$ is at most
$2\cdot n^{1+1/k_0}\left(\log n \cdot 8(\Delta^2+1)\right)^{k_0}\leq 2\cdot n^{1+2/k_0}.$

Therefore, the overall query complexity is at most
$3\cdot n^{1+2/k_0}$, which is $n^{1+O\left(\sqrt{(\log \log n+\log \Delta)/\log n}\right)},$
as stated in Theorem~\ref{thm:verify-bound-OPT}.

\begin{remark*}
The recursive algorithm for non-edge verification in this section (as well as the one in Section~\ref{sec:verify-bounded-width}) can be used for verification by itself.
However, we only use its query complexity to provide guarantee for the greedy algorithm in Section~\ref{sec:verif-set-cover}, because the greedy algorithm is much simpler.
\end{remark*}

\subsection{Proof of Theorem~\ref{thm:verify-bound-OPT-bounded-width}}
\label{sec:verify-bounded-width}
To show Theorem~\ref{thm:verify-bound-OPT-bounded-width}, we provide a recursive algorithm for non-edge verification of graphs of bounded treewidth with the query complexity in the theorem statement.
The algorithm first computes $(1/2)$-balanced separator in $\hat G$ and use it to obtain a partition of $V$.
Then it verifies the non-edges of $G$ between different components in the partition.
Finally, it recurses inside each component.
But there is a catch because of the query oracle: by querying a pair $(u,v)$, we would like to get back their distance in the recursive subgraph $H$, but instead the oracle returns their distance in the entire graph $G$.
It could well be that a shortest $u$-to-$v$ path in $G$ goes through two nodes $s_1$ and $s_2$ in the separator where the segment between $s_1$ and $s_2$ is outside $H$.

As a warmup, we first provide an algorithm for the special case of chordal graphs, because the above issue does not arise when the graph is chordal.\footnote{Since the separator is a clique, the shortest $s_1$-to-$s_2$ path is an edge, and thus belongs to $H$.}
We then extend the algorithm to graphs of bounded treewidth: To get around that issue, we formulate the recursive subproblem by augmenting $H$, adding virtual edges between vertices of the separator and giving them weight equal to their distance in $G$.

\subsubsection{Verifying Chordal Graphs.}
\label{sec:verif-chordal}
We have a recursive algorithm to verify that every non-edge of $\hat G$ is a non-edge of $G$ when $G$ is a chordal graph (Algorithm~\ref{alg:verify-chordal}).
The algorithm returns \emph{yes} if and only if every query during the execution gives the right distance.

\begin{algorithm}
\caption{Recursive Verification for Chordal Graphs}
\label{alg:verify-chordal}
\begin{algorithmic}[1]
\Procedure{Verify-Chordal}{$\hat G, U$}
    \If{$|U|>4(\Delta+1)$} 
    	\State $S\gets$ $(1/2)$-balanced clique separator of $\hat{G}[U]$ of size at most $\Delta+1$\label{def_clique_separator}
    	\State \textsc{Query}($S,U$) and obtain $N(S)\cap U$; \textsc{Query}($N(S)\cap U,U$)
    	\For{every component $C$ of $\hat{G}[U]\setminus S$}
        	\textsc{Verify-Chordal}$(\hat{G}, C\cup S)$
    	\EndFor
    \Else
    	 \;\textsc{Query}$(U,U)$
   	\EndIf
\EndProcedure
\end{algorithmic}
\end{algorithm}

Let $U\subseteq V$ represent the set of vertices for which we are currently verifying the induced subgraph.
By Corollary~\ref{cor:exist-clique-balanced}, there is a $(1/2)$-balanced clique separator $S$ of $\hat G[U]$.\footnote{We note that $S$ can be computed in polynomial time and with no queries.}
We confirm the non-edges between different components of $\hat G[U]\setminus S$ by querying every pair $(u,v)\in (N(S)\cap U)\times U$.
Then for each component $C$ of $\hat G[U]\setminus S$, we recursively verify the non-edges inside $\hat G[C \cup S]$.
The recursive call on the subset $C\cup S$ still use the global \textsc{Query} oracle.
But because $S$ is a clique in $G$, for any $u,v\in C\cup S$, any shortest $u$-to-$v$ path in $G$ stays inside $C\cup S$, so the value returned by  \textsc{Query}$(u,v)$ is the distance in $G[C\cup S]$.
The following lemma shows correctness of Algorithm~\ref{alg:verify-chordal} and is a main idea of the algorithm.

\begin{lemma}\label{lemma:check_neighbors}
Assume that $\hat E\subseteq E$.
If $\delta(u,v)= \hat{\delta}(u,v)$ for every $(u,v)\in (N(S)\cap U)\times U$, then there is no edge in $G[U]$ between different components of $\hat G[U]\setminus S$.
\end{lemma}

\begin{proof}
Let $X$ and $Y$ be any two different components in the partition of $\hat{G}[U]\setminus S$.
Let $x$ be any vertex in $X$ and $y$ be any vertex in $Y$.
We show that $\{x,y\}$ is not an edge in $G[U]$.
Let $a$ (resp. $b$) be the vertex in $N(S)$ that is closest to $x$ (resp. $y$) in $\hat{G}[U]$.
Then $a\in X$ and $b\in Y$.
Since $\hat E\subseteq E$, we have $\hat N(S)\subseteq N(S)$.
It is then easy to see that $a,\, b \, \in (N(S)\cap U) \setminus S$.
Without loss of generality, assume $\delta(a,x) \leq \delta(b,y)$.

Since $(a,y)\in (N(S)\cap U)\times U$, we have $\delta(a,y)=\hat \delta(a,y)$.
Any shortest path in $\hat G[U]$ from $a$ to $y$ goes through $S$, so
$$ \hat{\delta}(a,y)\geq \hat{\delta}(a,S)+\hat\delta(S,y)=\hat{\delta}(a,S)+1+\hat\delta(b,y)= 2+\hat\delta(b,y).$$
Since $(b,y)\in (N(S)\cap U)\times U$, we have $\hat \delta(b,y)= \delta(b,y)$.
Therefore $\delta(a,y)\geq 2+\delta(b,y)\geq 2+\delta(a,x)$. 
By the triangle inequality, $\delta(x,y)\geq \delta(a,y)-\delta(a,x)\geq 2$.
Thus $\{x,y\}$ is not an edge in $G[U]$.
\end{proof}

Since $\hat G[U]$ has maximum degree $\Delta$ and $S$ has size at most $\Delta+1$, \textsc{Query}($S,U$) and \textsc{Query}($N(S)\cap U,U$) use $O(\Delta^2|U|)$ queries.
Let $q(m)$ be the number of queries of \textsc{Verify-Chordal}($\hat{G},U$) when $|U|=m$.
We have
$$q(|U|)=O(\Delta^2|U|)+\sum_C q(|C|+|S|),$$
where $|U|=|S|+\sum_C |C|$ and $S$ is a $(1/2)$-balanced separator.
Hence $q(n)=O(\Delta^2n\log n)$.

\begin{remark*}
We note that there are simpler algorithms for verifying chordal graphs, but the algorithm presented here conveys ideas that can be extended to verify graphs of bounded treewidth.
\end{remark*}

\subsubsection{Verifying Graphs of Bounded Treewidth.}
\label{sec:verif-extention}
We extend Algorithm~\ref{alg:verify-chordal} to graphs of treewidth $\tw$.
The input specification is now the graph $\hat G$, a subset $U\subseteq V$, {\em plus} a set $F$ of additional, new edges $\{u,v\}$ with weight $\delta (u,v)$. 
The set $F$ is initially empty, and increases during the recursion.
The algorithm verifies whether the metric of $(U,\hat E[U]\cup F[U])$ is identical to that of $(U,E[U]\cup F[U])$.
Instead of $S$ being a clique, now $S$ is an existing bag of some tree decomposition of width $\tw$ (see Lemma~\ref{lemma:exist-balanced}).
Lemma~\ref{lemma:check_neighbors} still holds.
We create new edges $\{u,v\}$ with weight $w(u,v):=\delta (u,v)$ for all pairs $\{u,v\}\subseteq S$, and we add them to the set $F$. 
For each connected component $C$ of $\hat G[U]\setminus S$, we make a recursive call for the vertex set $C\cup S$ and the updated set $F$ of weighted edge.
Every subgraph in the recursive call has treewidth at most $\tw$, since the new edges are added inside $S$.
This concludes the description and correctness of the algorithm.

For the query  complexity, we need to bound the size of the neighborhood $N(S)$ of $S$: it is with respect to the subgraph $E[U]\cup F[U]$, so the vertex degree is no longer bounded by $\Delta$. 
However, for any vertex $v$, the number of weighted edges adjacent to $v$ is bounded by the maximum bag size times the number of bags $S$ containing $v$ that have been used as separators in the recursive calls. 
Since the graph has treewidth $\tw$, every bag has size at most $\tw+1$.
Since all separators are $(1/2)$-balanced, the recursion has depth $O(\log n)$, so $v$ belongs to $O(\log n)$ such bags.
Therefore, the degree of $v$ is $O(\Delta+\tw\log n)$. 
The overall query complexity is $O(\Delta(\Delta+\tw\log n)n\log n)$.

Thus we proved Theorem~\ref{thm:verify-bound-OPT-bounded-width}.

\section{Proof of Theorems~\ref{thm:recons-sp}}
\label{sec:recon}
The algorithm (Algorithm~\ref{alg:recons-greedy}) constructs an increasing set $X$ of edges so that in the end $X=E$.
At any time, the candidate graph is $X$.\footnote{We identify $X$ with the subgraph induced by the edges of $X$.}
Initially, $X$ is the union of the shortest paths given as answers by $n-1$ queries, so that $X$ is a connected subgraph spanning $V$.
At each subsequent step, the algorithm makes a query that leads either to the confirmation of many non-edges of $G$, or to the discovery of an edge of $G$.

Formally, we define, for every pair $(u,v)\in V^2$, 
\begin{equation}
\label{eqn:SRuv}
S^X_{u,v}=\big\{\{a,b\} \text{ is an non-edge of } X: \delta_X(u,a)+\delta_X(b,v)+1<\delta_X(u,v)\big\}.
\end{equation}
This is similar to $S_{u,v}$ defined in Equation~\eqref{eqn:Suv}.
From Lemma~\ref{lem:Suv}, $S^X_{u,v}$ contains the pairs that can be confirmed as non-edges of $G$ if $\delta_G(u,v)=\delta_X(u,v)$.
At each step, the algorithm queries a pair $(u,v)$ that maximizes the size of the set $S^X_{u,v}\setminus Y$.
As a consequence, either all pairs in $S^X_{u,v}\setminus Y$ are confirmed as non-edges of $G$, or $\delta_G(u,v)\neq \delta_X(u,v)$, and in that case, the query reveals an edge along a shortest $u$-to-$v$ path in $G$ that is not in $X$; we then add this edge to $X$.

\begin{algorithm}[t]
\caption{Greedy Reconstruction}
\label{alg:recons-greedy}
\begin{algorithmic}[1]
\Procedure{Reconstruct}{$V$}
    \State$u_0\gets$ an arbitrary vertex
    \For{$u\in V\setminus\{u_0\}$}
        \textsc{Query}$(u,u_0)$ to get a shortest $u$-to-$u_0$ path
    \EndFor
    \State $X\gets$ the union of the above paths
    \State$Y\gets\emptyset$
    \While{$X\cup Y$ does not cover all vertex pairs}
        \State choose $(u,v)$ that maximizes $|S^X_{u,v}\setminus Y|$ \Comment{$S^X_{u,v}$ defined in Equation~\eqref{eqn:SRuv}}\label{line:maximize}
        \State Query($u,v$) to get a shortest $u$-to-$v$ path \label{line:query}
        \If{$\delta_G(u,v)=\delta_X(u,v)$}
            \State $Y\gets Y\cup S^X_{u,v}$		\label{line:NE-update}
        \Else
            \State let $e$ be some edge of the above $u$-to-$v$ path that does not belong to $X$
            \State $X\gets X\cup \{e\}$			\label{line:add-edge}
        \EndIf
    \EndWhile
    \State \textbf{return} $X$
\EndProcedure
\end{algorithmic}
\end{algorithm}

To see the correctness, we note that the algorithm maintains the invariant that all pairs in $X$ are confirmed edges of $G$, and that all pairs in $Y$ are confirmed non-edges of $G$.
Thus when $X\cup Y$ covers all vertex pairs, we have $X=E$.

For the query complexity, first, consider the queries that lead to $\delta_G(u,v)\neq \delta_X(u,v)$.
For each such query, an edge is added to $X$.
This can happen at most $|E|\leq \Delta n$ times, because the graph has maximum degree $\Delta$.

Second, consider the queries that lead to $\delta_G(u,v)= \delta_X(u,v)$.
Define $R$ to be the set of vertex pairs that are not in $X\cup Y$.
We analyze the size of $R$.
For each such query, the size of $R$ decreases by $|S^X_{u,v}\setminus Y|$.
To lower bound $|S^X_{u,v}\setminus Y|$, we consider the problem of non-edge verification using a distance oracle on the input graph $X$, and let $T$ be an (unknown) optimal set of queries.
By Theorem~\ref{thm:verify-bound-OPT}, $|T|$ is at most $f(n,\Delta)=n^{1+O\left(\sqrt{(\log \log n+\log \Delta)/\log n}\right)}$.
By Lemma~\ref{lem:cover}, the sets $S^X_{u,v}$ for all pairs $(u,v)\in T$ together cover $R\cup Y$, hence $R$.
Therefore, at least one of these pairs satisfies $|S^X_{u,v}\setminus Y|\geq |R|/|T|$.
Initially, $|R|\leq n(n-1)/2$, and right before the last query, $|R|\geq 1$, thus the number of queries with $\delta_G(u,v)=\delta_X(u,v)$ is $O(\log n)\cdot f(n,\Delta)$.

Therefore, the overall query complexity is at most $(n-1)+\Delta n+O(\log n)\cdot f(n,\Delta)$.
Thus we obtained the same query bound as in the first statement of Theorem~\ref{thm:verify-dist}.
To prove the query bound for graphs of treewidth $\tw$ as in the second statement, the analysis is identical as above, except that we use Theorem~\ref{thm:verify-bound-OPT-bounded-width} instead of Theorem~\ref{thm:verify-bound-OPT} to obtain $f(n,\Delta)$.

\begin{remark*}
Note that the above proof depends crucially
on the fact that $f(n,\Delta)$ is a uniform bound on the number of distance queries for non-edge verification of {\em any} $n$-vertex graph of maximum degree $\Delta$. Thus, even though the graph $X$ changes during the course of the algorithm because of queries $(u,v)$ such that $\delta_G(u,v)\neq \delta_X(u,v)$, each query for which the distance in $G$ and the current $X$ are equal confirms $1/f(n,\Delta)$ fraction of non-edges.
\end{remark*}

\section{Proof of Theorem~\ref{thm:recons-chordal}}
\label{sec:recons-chordal}

The algorithm for Theorem~\ref{thm:recons-chordal} uses a clique separator to partition the graph into balanced subgraphs, and then recursively reconstructs each subgraph.
The main difficulty is to compute the partition.
The partition algorithm and its analysis are the main novelty in this section, see Section~\ref{sec:partition}.
In what follows, the set $U$ represents the set of vertices for which we are currently reconstructing the induced subgraph during the recursion.

\begin{definition}
A subset of vertices $U\subseteq V$ is said to be \emph{self-contained} if, for every pair of vertices $(x,y)\in U^2$, any shortest path in $G$ between $x$ and $y$ goes through nodes only in $U$.
\end{definition}

The set $U$ during the recursion is always \emph{self-contained}, because every separator is a clique.

\subsection{Subroutine: Computing the Partition}
\label{sec:partition}
Let $U$ be a self-contained subset of $V$.
Let $S$ be a subset of $U$.
We want to compute the partition of $G[U]\setminus S$ into connected components.
Let $W=(N(S)\cap U)\setminus S$.
For every $a\in W$, define $B(a)$ as the \emph{cluster} at $a$:
\begin{eqnarray}
\label{eqn:Ba}
B(a)=\{x\in U\setminus S\mid\delta(a,x)\leq \delta(S,x)\}.
\end{eqnarray}
Since $U$ is self-contained, every $x\in U\setminus S$ belongs to some cluster $B(a)$.
However, the clusters may have overlaps.
The algorithm (see Algorithm~\ref{alg:partition}) successively merges two clusters with overlaps.
See Figure~\ref{fig:partition}.

\begin{algorithm}
\caption{Computing the Partition}
\label{alg:partition}
\begin{algorithmic}[1]
\Function{Partition}{$U,S$}
    \State\textsc{Query}($S,U$) and obtain $N(S)\cap U$; \textsc{Query}($N(S)\cap U,U$)
    \State $W\gets (N(S)\cap U)\setminus S$
    \State $\mathcal{B}\gets\{B(a)\mid a\in W\}$  \Comment{$B(a)$ defined in Equation~\eqref{eqn:Ba}}
    \While{$\exists$ $B_1,B_2\in \mathcal{B}$ s.t. $B_1\cap B_2\neq \emptyset$}\label{merge_loop}
         merge $B_1$ and $B_2$ in $\mathcal{B}$
    \EndWhile
    \State \textbf{return} $\mathcal{B}$
\EndFunction
\end{algorithmic}
\end{algorithm}

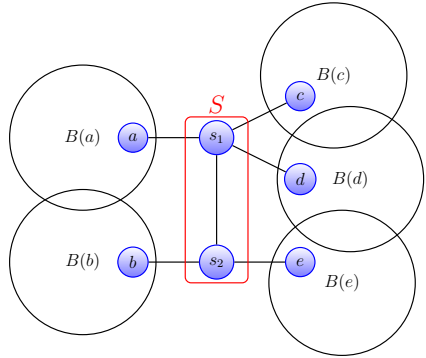
\tikzstyle{peers}=[draw,circle,blue,bottom color={blue!50},
                  top color= white, text=black,minimum width=20pt]
\begin{figure}[ht]
\begin{minipage}[c]{0.48\textwidth}
\centering
\scalebox{0.55}{
\large
\begin{tikzpicture}[auto, thick]
  \foreach \place/\name in {{(0,0)/b}, {(0,3)/a}, {(4,0)/e}, {(4,2)/d}, {(4,4)/c}}
    \node[peers] (\name) at \place {$\name$};
  \foreach \place/\name in {{(2,0)/2}, {(2,3)/1}}
    \node[peers] (s\name) at \place {$s_\name$};
  \foreach \source/\dest in {a/s1, b/s2, s1/c, s1/d,s2/e,s1/s2}
    \path (\source) edge (\dest);
  \node[rectangle, rounded corners, red, draw, minimum width=1.5cm, minimum height=4cm] at (2,1.5){};

  \node[circle, draw, minimum width=3.5cm] at (-1.2,0){$B(b)$};
  \node[circle, draw, minimum width=3.5cm] at (-1.2,3){$B(a)$};

  \node[circle, draw, minimum width=3.5cm] at (4+1,-0.5){$B(e)$};
  \node[circle, draw, minimum width=3.5cm] at (4+1.2,2){$B(d)$};
  \node[circle, draw, minimum width=3.5cm] at (4.8,4.5){$B(c)$};

  \node[] at (2,3.8) {\color{red} \LARGE $S$};
\end{tikzpicture}
}
\end{minipage}
\hfill
\begin{minipage}[c]{0.48\textwidth}
In the example, $S=\{s_1,s_2\}$ and $W=\{a,b,c,d,e\}$.
The clusters $B(a)$, $B(b)$, $B(c)$, $B(d)$, $B(e)$ are indicated by the balls. 
Using their overlaps, the algorithm produces the partition $\mathcal{B}=\{B(a)\cup B(b)\,,\,B(c)\cup B(d)\cup B(e)\}$.
\end{minipage}
\caption{Example of the Partition}
\label{fig:partition}
\end{figure}

\begin{lemma}
\label{lem:partition}
Algorithm \textsc{Partition} uses $O(\Delta|S|\cdot|U|)$ queries and outputs the partition of $G[U]\setminus S$ into connected components.
\end{lemma}

The query complexity of the algorithm is $O(|N(S)|\cdot|U|)= O(\Delta|S|\cdot|U|)$.
Lemma~\ref{lem:partition} then follows directly from Lemmas~\ref{lemma:connected1}~and~\ref{lemma:connected2}.

\begin{lemma}
\label{lemma:connected1}
Let $C$ be a connected component in $G[U]\setminus S$.
Then $C\subseteq B$ for some set $B$ in the output of the algorithm.

\end{lemma}

\begin{proof}
Let $A$ be the set of vertices in $C \cap W$.
Since $U$ is self-contained, for every vertex $x\in C$, there exists some $a\in A$ such that $x\in B(a)$.
Thus we only need to prove that all sets $\{B(a): a\in A\}$ are eventually merged in our algorithm.

Define a weighed graph $H$ whose vertex set is $A$, and such that for every $(a,b)\in A^2$, there is an edge $(a,b)$ in $H$ with weight $w(a,b)$, which is defined as the distance between $a$ and $b$ in $G[C]$\footnote{This distance may be larger than $\delta(a,b)$, the distance between $a$ and $b$ in $G$.}.
To show that all sets $\{B(a): a\in A\}$ are eventually merged, we use an inductive proof that is in the same order that Prim's algorithm would construct a minimum
spanning tree on $H$.
Recall that Prim's algorithm initializes a tree $\mathcal{T}$ with a single vertex, chosen arbitrarily from $A$.
Then it repeatedly chooses an edge $(a,b)\in \mathcal{T}\times(A\setminus \mathcal{T})$ with minimum weight and add this edge to $\mathcal{T}$.
We will show that if an edge $(a,b)$ is added to $\mathcal{T}$, then $B(a)$ and $B(b)$ are merged in our algorithm.
Since Prim's algorithm finishes by providing a spanning tree including every $a\in A$, we thus proved that all sets $B(a)$ for $a\in A$ are merged in our algorithm.

Suppose that the $i$ unions corresponding to the first $i$ edges chosen by Prim's algorithm have been performed already, for $i\geq 0$.
Let $\mathcal{T}$ be the tree in $H$ after adding the first $i$ edges.\footnote{For the base case ($i=0$), $\mathcal{T}$ contains a single vertex and no union operation is performed.}
Let $(a,b)$ be the $(i+1)^{\rm th}$ edge chosen by Prim's algorithm.
Thus $a\in \mathcal{T}$, $b\in A\setminus \mathcal{T}$, and $w(a,b)$ is minimized.
Consider a shortest path $p_1,\dots, p_k$ in $G[C]$ between $a$ and $b$.
Let $z=p_{\lceil k/2\rceil}$ be the mid-point vertex of the path.
We show that both $B(a)$ and $B(b)$ contain $z$, thus $B(a)$ and $B(b)$ are merged in our algorithm.
It is easy to see that $p_1,\dots,p_{\lceil k/2\rceil}$ and $p_{\lceil k/2\rceil},\dots,p_k$ are shortest paths in $G$.
Thus $\delta(a,z)=\lceil k/2 \rceil -1$ and $\delta(b,z)=\lfloor k/2\rfloor$.
So we have $\delta(a,z)\leq \delta(b,z)\leq \delta(a,z)+1$.
To show $z\in B(a)$ and $z\in B(b)$, we only need to show that $ \delta(b,z)\leq \delta(S,z)$.
Choose the vertex $s\in S$ that minimizes $\delta(s,z)$ and consider a shortest $z$-to-$s$ path $P$.
Let $c$ be the neighbor of $s$ on $P$, and let $P'$ be the shortest $z$-to-$c$ path.
We note that $c\in A$ and $P'$ is in $G[C]$. 
Since $\delta(S,z)=\delta(s,z)=\delta(c,z)+1$, we only need to show that $\delta(b,z)\leq \delta(c,z)+1$.
There are 2 cases:

Case 1: $c\in A\setminus \mathcal T$. Then the concatenation of $p_1,\dots,p_{\lceil k/2\rceil}$ and $P'$ gives a path in $G[C]$ between $a$ and $c$ of length $\delta(a,z)+\delta(c,z)$, which is at least $w(a,c)$ by the definition of the weight.
From the choice of $(a,b)$, $w(a,c)\geq w(a,b)=\delta(a,z)+\delta(b,z)$.
So we have $\delta(b,z)\leq \delta(c,z)$.

Case 2: $c\in \mathcal T$. Similarly, the concatenation of $p_k,p_{k-1},\dots,p_{\lceil k/2\rceil}$ and $P'$ gives a path in $G[C]$ between $b$ and $c$ of length $\delta(b,z)+\delta(c,z)$, which is at least $w(b,c)$ by the definition of the weight.
From the choice of $(a,b)$, $w(b,c)\geq w(a,b)=\delta(a,z)+\delta(b,z)$.
So we have $\delta(a,z)\leq \delta(c,z)$.
Thus $\delta(b,z)\leq \delta(a,z)+1\leq \delta(c,z)+1$.
\end{proof}

\begin{lemma}
\label{lemma:connected2}
Let $B$ be a set in the output of the algorithm. Then $B\subseteq C$ for some connected component $C$ in $G[U]\setminus S$.
\end{lemma}

\begin{proof}
First we show that for every $a\in W$ and every $x\in B(a)$, $a$ and $x$ belong to the same component in $G[U]\setminus S$.
Suppose there exists some $x\in B(a)$, such that $x$ and $a$ belong to different components in  $G[U]\setminus S$.
Any shortest path from $a$ to $x$ must pass through the separator $S$, so we have $\delta(a,x)\geq \delta(a,S)+\delta(S,x)=1+\delta(S,x)$.
Contradiction with $x\in B(a)$.

Next we prove an invariant on $\mathcal{B}$ during the \textbf{while} loop (Line~\ref{merge_loop}):
\emph{Every set $B\in\mathcal{B}$ is a subset of some component of $G[U]\setminus S$.}
This invariant holds before the \textbf{while} loop starts.
Suppose the invariant holds before the $i^{\rm th}$ iteration of the \textbf{while} loop, and in this iteration $B_1, B_2\in\mathcal{B}$ get merged.
Since $B_1\cap B_2\neq \emptyset$, there exists $z\in B_1\cap B_2$.
All nodes in $B_1$ (resp. in $B_2$) are in the same component as $z$.
Thus all nodes in $B_1\cup B_2$ are in the same component as $z$.
By induction, the invariant holds when the \textbf{while} loop terminates.

Thus we complete the proof.
\end{proof}

\subsection{Subroutine: Computing a Shortest Path}
\label{sec:shortest path}

Given a self-contained subset of vertices $U\subseteq V$ and two vertices $a,b\in U$, Algorithm~\ref{alg:shortest} computes a shortest path between $a$ and $b$ by divide-and-conquer.
The query complexity is $O(|U|\log |U|)$.
See Appendix~A.1 of~\cite{mathieu2013graph} for the analysis of the algorithm.

    \begin{algorithm}[t]
    \caption{Finding a Shortest Path}
    \label{alg:shortest}
    \begin{algorithmic}[1]
    \Function{Shortest-Path}{$U,a,b$}
        \If{$\delta(a,b)>1$}
        \State\textsc{Query}($a,U$); \textsc{Query}($b,U$)
        \State $T\gets\{v\in U\mid \delta(v,a)+\delta(v,b)=\delta(a,b)\}$
        \State $l\gets \lfloor\delta(a,b)/2\rfloor$
        \State $c\gets$ an arbitrary node in $T$ such that $\delta(c,a)=\ell$
        \State $U_1\gets\{v\in T \mid \delta(v,a)<\ell\}$
        \State $U_2\gets\{v\in T \mid \delta(v,a)>\ell\}$
        \State $P_1\gets\textsc{Shortest-Path}(U_1,a,c)$
        \State $P_2\gets\textsc{Shortest-Path}(U_2,c,b)$
        \State \textbf{return} the concatenation of $P_1$ and $P_2$
        \Else
            \State\textbf{return} the path of a single edge $(a,b)$
        \EndIf
    \EndFunction
    \end{algorithmic}
    \end{algorithm}

\subsection{Algorithm and Analysis}
\label{sec:recons-chordal-algo}
The reconstruction algorithm is in Algorithm~\ref{alg:recons-chordal}.
To find a balanced separator, we use ideas from~\cite{mathieu2013graph}: the algorithm computes a vertex that is on many shortest paths in the sampling, and grows a clique including this vertex.
The constants $n_0$, $C_1$, and $0<\beta<1$ are defined later.

\begin{algorithm}
\caption{Reconstruction of Chordal Graphs}
\label{alg:recons-chordal}
\begin{algorithmic}[1]
\Procedure{Reconstruct-Chordal}{$U$}
    \If{$|U|> n_0$}
        \State $K\gets$\textsc{Balanced-Separator}$(U)$
        \State $(U_1,\dots,U_\ell)\gets$\textsc{Partition}$(U,K)$ \Comment{See Algorithm~\ref{alg:partition}}
        \State \textbf{return} $\bigcup_i$\textsc{Reconstruct-Chordal}$(U_i\cup K)$
    \Else
        \State reconstruct $G[U]$ by \textsc{Query}$(U,U)$
    \EndIf
\EndProcedure
\Function{Balanced-Separator}{$U$}  \Comment{finds a $\beta$-balanced separator of $G[U]$}
\Repeat
        \For{$i\gets 1$ \textbf{to} $C_1\log |U| $} 
            \State $(a_i,b_i)\gets$ a pair of  uniformly random nodes from $U$
            \State $P_i\gets\textsc{Shortest-Path}(a_i,b_i,U)$ \label{shortest_path_loop_end}\Comment{see Section~\ref{sec:shortest path}}
        \EndFor
        \State $x\gets$ the node in $U$ with the most occurrences among all $P_i$'s \label{def_x}
        \State\textsc{Query}$(x,U)$ and obtain $N(x)$
        \State\textsc{Query}($N(x),N(x)$) and obtain all cliques containing $x$
        \For{every clique $K$ containing $x$}\label{check_clique}
            \State$(U_1,\dots,U_\ell)\gets\textsc{Partition}(U,K)$ \Comment{See Algorithm~\ref{alg:partition}}
            \If{$\max_i |U_i|<\beta |U|$}
                \textbf{return} $K$
            \EndIf
        \EndFor
    \Until{a balanced separator is found}
\EndFunction
\end{algorithmic}
\end{algorithm}

\begin{lemma}
\textsc{Reconstruct-Chordal}$(U)$ indeed returns the edge set of $G[U]$.
\end{lemma}
\begin{proof}
By Lemma~\ref{lem:partition}, $U_1,\dots,U_\ell$ are the connected components in $G[U]\setminus K$.
There cannot be edges between different $U_i$ and $U_j$.
Thus every edge of $G[U]$ belongs to some $G[U_i\cup K]$.
So the edge set of $G[U]$ is the union of the edge sets of $G[U_i\cup K]$ over $i$.
Hence correctness follows by induction.
\end{proof}

The rest of this section is to analyze the query complexity.
We set the constants $n_0=2^{\Delta+2}(\Delta+1)^2$;  $\beta=\max\left(1-1/(\Delta\cdot 2^{\Delta+1}),\sqrt{1-1/(4(\Delta+1))}\right)$; and $C_1=256(\Delta+1)^2$.
The key is the following lemma.

\begin{lemma}
\label{lem:success-proba}
In every \textbf{repeat} loop of \textsc{Balanced-Separator}, a $\beta$-balanced separator is found with probability at least $2/3$.
\end{lemma}

We defer the proof of Lemma~\ref{lem:success-proba} to Section~\ref{sec:proof-success-proba} and show in the rest of this section how Lemma~\ref{lem:success-proba} implies the query complexity stated in Theorem~\ref{thm:recons-chordal}.

First we analyze the query complexity of \textsc{Balanced-Separator}.
Computing $C_1\log |U|$ shortest paths takes $O(\Delta^2|U|\log^2|U|)$ queries, since a shortest path between two given nodes can be computed using $O(|U|\log|U|)$ queries (see Section~\ref{sec:shortest path}).
We note that the neighborhood $N(x)$ of $x$ has size at most $\Delta+1$, and there are at most $2^\Delta$ cliques containing $x$.
By Lemma~\ref{lem:partition}, \textsc{Partition}$(U,K)$ takes $O(\Delta|K|\cdot|U|)$ queries, where $|K|\leq \Delta+1$.
Therefore every \textbf{repeat} loop in \textsc{Balanced-Separator} takes $O\left(\Delta^2|U|(2^\Delta+\log^2|U|)\right)$ queries.
By Lemma~\ref{lem:success-proba}, the expected number of \textbf{repeat} loops is constant.
So the query complexity of \textsc{Balanced-Separator} is $O\left(\Delta^2|U|(2^\Delta+\log^2|U|)\right)$.

Next, we analyze the query complexity of \textsc{Reconstruct-Chordal}$(U)$.
Let $q(m)$ be the number of queries when $|U|=m$.
We have 
$$q(|U|)=O\left(\Delta^2|U|(2^\Delta+\log^2|U|)\right)+\sum_i q(|U_i|+|K|),$$
where $|U|=|K|+\sum_i |U_i|$ and $K$ is a $\beta$-balanced separator of size at most $\Delta+1$.
Hence $q(n)=O\left(\Delta^2n(2^\Delta+\log^2n)\log_{\frac{1}{\beta}}n\right)=O\left(\Delta^3 2^\Delta\cdot n(2^\Delta+\log^2n)\log n\right)$.

\subsection{Proof of Lemma~\ref{lem:success-proba}}
\label{sec:proof-success-proba}
First, we need Lemmas~\ref{lemma:most popular node} and~\ref{lemma:sample accuracy}.
\begin{lemma}
\label{lemma:most popular node}
For $v\in U$, let $p_v$ denote the fraction of pairs $(a,b)\in U^2$ such that $v$ is on some shortest path between  $a$ and $b$.
Then $\max_v p_v \geq 1/(2(\Delta+1))$.
\end{lemma}

\begin{proof}
By Corollary~\ref{cor:exist-clique-balanced}, there is some clique separator $S$ of size at most $\Delta+1$ such that every connected component in $G[U]\setminus S$ has size at most $|U|/2$.
Notice that for any pair of vertices $a,b$ from different components, any shortest $a$-to-$b$ path must go by some node in $S$.
The number of such pairs is at least $|U|^2/2$.
By Pigeonhole Principle, there exists some $z\in S$, such that for at least $1/ |S|\geq 1/(\Delta+1)$ fraction of these pairs, their shortest paths go by $z$.
Thus $p_z\geq 1/(2(\Delta+1))$.
\end{proof}

\begin{lemma}[slightly adapted from~\cite{mathieu2013graph}]
\label{lemma:sample accuracy}
For ever vertex $v\in U$, let $\hat{p}_v$ denote the fraction of pairs $(a_i,b_i)$ among $C_1\log |U|$ uniformly and independently random pairs of $U^2$ such that $v$ is on some shortest path between $a_i$ and $b_i$.
Let $x=\arg\max \hat{p}_x$.
Then with probability at least $2/3$, we have $p_x>(\max_v p_v)/2$.
\end{lemma}

Now we prove Lemma~\ref{lem:success-proba}.
By Lemma~\ref{lemma:chordal}, there is a tree decomposition $T$ of $G[U]$ such that every bag of $T$ is a unique maximal clique of $G[U]$.
Let $x$ be the node computed on Line~\ref{def_x} of Algorithm~\ref{alg:recons-chordal}.
Let $T_x$ be the subtree of $T$ induced by the bags containing $x$.
Define $F$ to be the forest after removing $T_x$ from $T$.
For any subgraph $H$ of $T$, define $V(H)\subseteq U$ to be the set of vertices that appear in at least one bag of $H$.

\emph{Case 1:}
There exists some connected component $T'$ in $F$ with $(1-\beta)|U|\leq |V(T')|\leq \beta |U|$.
Consider the edge $(K_1,K_2)$ in $T$ such that $K_1\in T_x$ and $K_2\in T'$.
$K_1\cap K_2$ is a $\beta$-balanced separator, since $V(T')$ is a component in $G[U]\setminus(K_1\cap K_2)$.
Thus $K_1\supseteq K_1\cap K_2$ is also a $\beta$-balanced separator.
Observe that $x\in K_1$, so $K_1$ is one of the cliques checked on Line~\ref{check_clique}.
The algorithm succeeds by finding a $\beta$-balanced separator.

\emph{Case 2:}
There exists some connected component $T'$ in $F$ with $|V(T')|>\beta |U|$.
The algorithm then fails to find a $\beta$-balanced separator.
We bound the probability of this case by at most $1/3$.
Again let $(K_1,K_2)$ be the edge in $T$ such that $K_1\in T_x$ and $K_2\in T'$.
For any vertices $u,v\in V(T')$, any shortest $u$-to-$v$ path cannot go by $x$.
Since there are at least $\beta^2$ fraction of such pairs in $U^2$, we have $p_x\leq 1-\beta^2$, which is at most $1/(4(\Delta+1))$ by the definition of $\beta$.
This happens with probability at most $1/3$ by Lemmas~\ref{lemma:most popular node} and \ref{lemma:sample accuracy}.

We argue that the two cases above are exhaustive.
Suppose, for the sake of contradiction, that every component $T'$ in $F$ is such that $|V(T')|<(1-\beta)|U|$.
The number of components in $F$ is at most $\Delta\cdot 2^\Delta$, because every component has a bag that contains a neighbor of $x$, and all bags are unique.
So $|V(F)|< \Delta\cdot 2^\Delta \cdot (1-\beta)|U|$, which is at most $|U|/2$ by the definition of $\beta$.
On the other hand, every node $v\in U\setminus N(x)$ is covered by some clique in $F$, so $|V(F)|\geq |U|-(\Delta+1)$, which is greater than $|U|/2$ since $|U|> n_0$.
Contradiction.

Thus we complete the proof of Lemma~\ref{lem:success-proba}.

\section{Lower Bounds}

\subsection{Lower Bound for Graphs of Unbounded Degree}
\label{sec:lower-bound-unbounded}
Reconstruction of graphs of unbounded degree using a distance oracle requires $\Omega(n^2)$ queries~\cite{reyzin2007}.
This lower bound can be easily extended to verification or/and to the shortest path oracle model as follows.
Consider the graph $G$ of vertices $v_1,\dots,v_n$, which contains a star: it has an edge $\{v_1,v_i\}$ for every $2\leq i\leq n$.
$G$ may or may not contain one additional edge $\{v_i,v_j\}$ for $2\leq i, j\leq n$.
(In the verification version of the problem, the star graph is given as $\hat G$.)
To detect if $G$ contains such an edge $\{v_i,v_j\}$ for $2\leq i,j\leq n$, we need to perform $\Omega(n^2)$ distance or shortest path queries.

\subsection{Lower Bound for Reconstruction of Bounded Degree Graphs}
\label{sec:information-lower-bound}
We assume that $n=3 t-1$ where $t=2^k$ for some integer $k$. (The general case is similar.)
Consider a family $\mathcal{G}$ of graphs $G$ as follows: the vertex set is $\{v_1,\dots,v_n\}$; the first $2t-1$ vertices form a complete binary tree of height $k$ (with leaves $v_t,\dots,v_{2t-1}$); the other vertices $v_{2t},\dots,v_{3t-1}$ induce an arbitrary subgraph of maximum degree $\Delta-1$; there is an edge between $v_i$ and $v_{i+t}$ for every $i\in[t,2t-1]$ and there are no other edges.
Then every vertex in $G$ has degree at most $\Delta$, and the diameter of the graph is at most $2k+2$.
Every distance query returns a number between 1 and $2k+2=O(\log n)$, so it gives $O(\log\log n)$ bits of information.
From information theory, the number of queries is at least the logarithm of the number of graphs in $\mathcal{G}$ divided by the maximum number of bits of information per query.
The number of graphs in $\mathcal{G}$ is the number of graphs of size $t$ and of maximum degree $\Delta-1$, which is $\Omega\left(n^{\Omega(\Delta n)}\right)$ when $\Delta=o(\sqrt{n})$ (see \cite{mckay1991asymptotic}).
Therefore, we have a query lower bound of \[\frac{\log \left(\Omega\left(n^{\Omega(\Delta n)}\right)\right)}{O(\log\log n)}=\Omega\left(\frac{\Delta n\log n}{\log\log n}\right).\]

\subsubsection*{Acknowledgments.}
We thank Uri Zwick for Theorem~\ref{thm:information-lower-bound}.
We thank Fabrice Benhamouda, Mathias Bæk Tejs Knudsen, and Mikkel Thorup for discussions.

\bibliographystyle{splncs03}
\bibliography{reference}
\end{document}